\RequirePackage{amsmath}
\documentclass{llncs}
\usepackage{makeidx}  
\usepackage{tkz-euclide,tkz-fct}
\usetkzobj{all}
\usepackage[fleqn,tbtags]{mathtools}
\usepackage{amsfonts}%
\usepackage{amssymb}%
\usepackage{wasysym}
\usepackage{graphicx}
\usepackage{pgf}
\usepackage{tabularx}

\usepackage{import}
\usepackage{filecontents}

\usepackage[titlenumbered,linesnumbered,ruled,noend,algosection]{algorithm2e}

\usepackage{float}
\usepackage{ifpdf}
\usepackage{fancyvrb}
\usepackage{relsize}

\usepackage{verbatim}

\usepackage{lastpage}
\usepackage{bookmark}
\bookmarksetup{numbered}
\usepackage{hyperref}
\usepackage[capitalise]{cleveref}
\usepackage[T1]{fontenc} 
\usepackage[inline]{enumitem}

\newcommand{\conv}{\operatornamewithlimits{conv}}

\newcommand{\CALL}[1]{\textsc{#1}}

\newcommand{\BigO}[1]{\ensuremath{\mathcal{O}\left(#1\right)}}
\newcommand{\BigTheta}[1]{\ensuremath{\Theta\left(#1\right)}}

\newcommand{\head}{\operatornamewithlimits{\CALL{head}}}
\newcommand{\tail}{\operatornamewithlimits{\CALL{tail}}}
\newcommand{\pred}{\operatornamewithlimits{\CALL{pred}}}
\newcommand{\lsucc}{\operatornamewithlimits{\CALL{succ}}}

\newcommand{\nil}{\text{\textbf{nil}}}
\newcommand{\LAnd}{\text{\textbf{ and }}}
\newcommand{\Or}{\text{\textbf{ or }}}

\DontPrintSemicolon
\SetKwComment{Comment}{$\rhd \ $}{$\lhd$}
\SetKwComment{tcp}{$\rhd \ $}{$\lhd$}

\SetCommentSty{mycommfont}
\crefname{section}{Sect.}{Sect.}

\newcommand{\tangents}{\operatornamewithlimits{\CALL{tangents}}}

\newcommand{\BuildTree}{\operatornamewithlimits{\CALL{buildTree}}}

\newcommand{\extractHull}{\operatornamewithlimits{\CALL{extractHull}}}

\newcommand{\Use}{\operatornamewithlimits{\mathit{Use}}}
\newcommand{\aabove}{\operatornamewithlimits{\CALL{above}}}
\newcommand{\bbelow}{\operatornamewithlimits{\CALL{below}}}

\newcommand{\getBridge}{\operatornamewithlimits{\CALL{getBridge}}}

\newcommand{\iinsert}{\operatornamewithlimits{\CALL{insert}}}

\newcommand{\delete}{\operatornamewithlimits{\CALL{delete}}}

\newcommand{\merge}{\operatornamewithlimits{\CALL{merge}}}

\begin{document}


\title{A Simple Convex Layers Algorithm}
\titlerunning{Convex Layers}  
%
\author{Raimi A. Rufai\inst{1} \and Dana S. Richards \inst{2}}
\authorrunning{Raimi Rufai and Dana Richards} 
%
\tocauthor{Raimi Rufai and Dana Richards}
\institute{SAP Labs, 111 rue Duke, Suite 9000, Montreal QC H3C 2M1, Canada\\
\email{raimi.rufai@sap.com}\\ 
\and
Department of Computer Science, George Mason University, \\
4400 University Drive MSN 4A5, Fairfax, VA 22030, USA \\
\email{richards@cs.gmu.edu}}

\maketitle 

\begin{abstract}
Given a set of $n$ points $P$ in the plane, the first layer $L_1$ of $P$ is formed by the points that appear on $P$'s convex hull. In general, a point belongs to layer $L_i$, if it lies on the convex hull of the set $P \setminus \bigcup_{j<i}\{L_j\}$. The \emph{convex layers problem} is to compute the convex layers $L_i$. Existing algorithms for this problem either do not achieve the optimal $\BigO{n\log n}$ runtime and linear space, or are overly complex and difficult to apply in practice. We propose a new algorithm that is both optimal and simple. The simplicity is achieved by independently computing four sets of monotone convex chains in $\BigO{n\log n}$ time and linear space. These are then merged in  $\BigO{n\log n}$ time.
\end{abstract}

\section{Introduction}
\label{section:layering:intro}
Algorithms for the convex layers problem that achieve optimal time and space complexities are arguably complex and discourage implementation. 
We give a simple $\BigO{n \log n}$-time and linear space algorithm for the problem, which is optimal. 
Our algorithm computes four quarter convex layers using a plane-sweep paradigm as the first step. The second step then merges these  in $\BigO{n \log n}$-time.

Formally, the \emph{convex layers}, $\mathbb{L}(P) = \{L_1, L_2,
\cdots, L_k \}$, of a set $P$ of $n \geq 3$ points is a partition of $P$ into $k
\leq \lceil n/3\rceil$ disjoint subsets $L_i$, $i = 1, 2, \cdots, k$ called
\emph{layers}, such that each layer $L_i$ is the ordered\footnote{That is the
layers are polygons not sets.} set of vertices of the convex hull of $P \setminus \bigcup_{j<i}\{L_j\}$.
Thus, the outermost layer $L_1$ coincides exactly with the convex hull
of $P$, $\conv(P)$. The \emph{convex layers problem} is to
compute $\mathbb{L}(P)$.

\label{section:layering:applications}
Convex layers have several applications in various domains, including robust
statistics, computational geometry, and pattern recognition. 

\section{Related Work}
\label{section:layering:relatedWorks} 
A brute-force solution to the convex layers problem is obvious---construct each
layer $L_i$  as the convex hull of the set $P\setminus \bigcup_{j < i} L_j$ using some
suitable convex hull algorithm. 
The brute-force algorithm will take $O(kn\log n)$ time where $k$ is the number of the layers. 
We say this algorithm ``peels off'' a set of points one layer at a time.
This \emph{peeling} approach is reminiscent of many convex layers algorithms. 
Another general approach to this problem is the \emph{plane-sweep} paradigm. 

One of the earliest works to take the peeling approach is Green and Silverman
\cite{green79}. Their algorithm repeatedly invokes quickhull to extract
a convex layer at each invocation. This algorithm runs in $O(n^3)$
worst-case time, and $O(n^2\log n)$ expected time.

Overmars and van Leeuwen's \cite{Overmars1981166} algorithm for this
problem runs in $O(n\log^2 n)$. It is based on a fully dynamic data structure for
maintaining a convex hull under arbitrary insertions and deletions of points.
Each of these update operations takes $O(\log^2 n)$ time, since constructing the
convex layers can be reduced to inserting all the points into the data structure
in time $O(n\log^2 n)$, marking points on the current convex hull, deleting
them, and then repeating the process for the next layer. Since each point is marked
exactly once and deleted exactly once, these steps
together take no more than $O(n\log^2 n)$ time. 

Chazelle's \cite{chazelle85} algorithm for this problem runs in $O(n\log n)$ time and $O(n)$ space, 
both of which are optimal. A new algorithm is discussed in \cref{section:layering:results}. (Chazelle \cite{chazelle85} used a balanced tree approach as well as our new algorithm, 
but the information stored in our tree corresponds to a very different set of polygonal chains.)

The first algorithm on record that uses the plane-sweep approach is a modification of
Jarvis march proposed by Shamos \cite{ps-cgi-90}. The algorithm works by doing
a radial sweep, changing the pivot along the way, just as in the Jarvis march,
but does not stop after processing all the points. It proceeds with another
round of Jarvis march that excludes points found to belong to the convex hull
in the previous iteration. This way, the algorithm runs in $O(n^2)$.

Nielsen \cite{nielsen96} took advantage of Chan's grouping trick \cite{chan96b} to obtain yet another optimal 
algorithm for the convex layers problem. Nielsen's algorithm is output-sensitive in that it can be parametrized
by the number of layers $l$ to be computed. It runs in $O(n\log H_l)$ time where $H_l$
is the number of points appearing on the first $l$ layers.

\section{New Algorithm}
\label{section:layering:results}
Our algorithm builds four sets of convex layers, each with a distinct 
direction of curvature.  
The set of points $P$ must be known ahead of time.
For ease of presentation, we will assume below that the points are in general position 
(no three on a line and no two share a coordinate). Removing this assumption is easy and our implementation does not make the assumption. Each point's horizontal ranking is precomputed by sorting the points using their $x$-coordinates.

A \emph{northwest} monotone convex chain $C = (p_1, p_2, \cdots, p_n)$ has
$p_i.x<p_{i+1}.x$, $p_i.y<p_{i+1}.y$ and no point in the chain is above the extended line 
through $p_i$ and $p_{i+1}$, $1\le i < n$. The head (tail) of the chain $C$ is defined as $\head(C)=p_1$ ($\tail(C)=p_n$).
A {\em full} monotone convex chain is formed by augmenting $C = (p_1, p_2, \cdots, p_n)$ 
with two fictional \emph{sentinel} points, $p_0 = (p_{1}.x, -\infty)$, and $p_{n+1} = (\infty, p_n.y)$. Note that given  a full chain $C$, the calls $\head(C)$ and $\tail(C)$ return $p_1$ and $p_n$ and not the sentinels. The $\infty$'s are symbolic and the operations are overloaded for them.

We call the chain \emph{northwest} since it bows outward in that direction.
Similarly we will refer to a chain as \emph{southwest} if it is northwest  after rotating  the point set
90 degrees clockwise about the origin. \emph{Northeast} and \emph{southeast} are defined analogously.
Except in the section on merging below, {\bf all} of our chains will be northwest monotone 
convex chains, so we will simply call them \emph{chains}.

We say a chain $C_1$ \emph{precedes} another chain $C_2$, if $\tail(C_1).x < \head(C_2).x$.
Additionally, we say a line is \emph{tangent} to a chain, if it touches the chain and no point in the chain is above the line.

Let chain $C_1$ precede chain $C_2$.
If $\tail(C_1).y < \tail(C_2).y$ then a \emph{bridge} between the chains is a two-point chain
$(p_1, p_2)$ where $p_1$ is in chain $C_1$, $p_2$ is in chain $C_2$ and 
the line through $p_1$ and $p_2$ is tangent to both chains.
If $\tail(C_1).y \ge \tail(C_2).y$ then the chain $(\tail(C_1), \tail(C_2))$ is a \emph{degenerate bridge}.

Let $C = (p_0, p_1, p_2, \cdots, p_n, p_{n+1})$ be a full chain.
$C$ \emph{dominates} a point $p$ if $p$ is below the segment $(p_i, p_{i+1})$, for some
$1\le i\le n$.
A full chain $C$ dominates a full chain $C'$ if $C$ dominates every point of $C'$. The (northwest) \emph{hull} 
of a point set $P$, or just the \emph{hull chain} of $P$, is the chain of points 
from $P$ that dominates every other point in $P$.

\subsection{Hull Tree Data Structure}
The \emph{hull tree} $T$ for the point set $P$, is a binary tree. $T$ is either $\nil$ or has
\begin{enumerate*}[label=\emph{\alph*}), before=\unskip{: }, itemjoin={{; }}, itemjoin*={{, and }}]\item
A root node that contains the hull chain for $P$
\item
A partition of the non-hull points from $P$ by some $x$ coordinate into $P_L$ and $P_R$.
The left and right children of the root are the hull trees $T.L$ and $T.R$ for $P_L$ and $P_R$,
respectively. 
\end{enumerate*}

The root node contains the  fields
\begin{enumerate*}[label=\emph{\alph*}), before=\unskip{: }, itemjoin={{; }}, itemjoin*={{, and }}]
\item
$T.hull$ is the full hull chain
\item
$T.l$ is a cursor into the hull chain, that initially scans rightwards 
\item
$T.r$ is a cursor into the hull chain, that initially goes leftwards.
\end{enumerate*}

The reason for the two cursors is so that we can be explicit about how the hull chain is scanned.
Our  analysis depends on the claim that a point is only scanned a constant number of times before being deleted from that chain. 
We will maintain the invariant that $T.l$ is never to the right of $T.r$.
(For example, if $T.r$ coincides with $T.l$ and $T.r$ moves left then $T.l$ 
will be pushed back to maintain the invariant.)

As a preprocessing step, the points are sorted by $x$-coordinates and assigned a 0-based
integer rank, represented by a $\lceil\log n\rceil$ bit binary string.
We will use a perfectly balanced binary tree as the skeleton for our hull tree. 
It will have $n$ leaves with no gaps from left to right. The skeleton is only implicit in our construction.
The leaves will be associated with the points by rank from left to right.
We will use the conventional trick that the path from the root to a leaf is determined by the
binary representation of the rank of the leaf, going left on 0 and right on 1.
We now specify $P_L$ as the set of points, not in the hull chain, 
whose rank starts with 0 in binary and $P_R$ is analogous with a leading bit of 1.
We will say a point in, say, $P_L$ ``belongs'' to the subtree $T.L$.

As hull points are peeled off, the corresponding leaf nodes will become obsolete, but we never
recalculate ranks.  It follows that the skeleton of the tree will always be of height $\lceil\log n\rceil$,
even though, over time, more and more of the bottom of the tree will become vacant.
In addition to this invariant, we explicitly mention these other invariants:
(a) $T.hull$ is a full northwest monotone convex chain, (b) $T.hull$ is the hull of $P$, 
(c) $T.l.x \le T.r.x$, and (d) $T.L$ and $T.R$ are the hull trees for $P_L$ and $P_R$.

\begin{lemma}
The space complexity of a hull-tree $T$ for a set $P$ of $n$ points is $\BigTheta{n}$.
\end{lemma}

\begin{proof}
The skeleton of the binary tree with $n$ leaves clearly has $O(n)$ nodes of constant size.
There is also the space for all the lists representing the various hull chains.
However, from the definitions of $P_L$ and $P_R$, every point is on exactly one hull chain.
This completes the proof. \qed
\end{proof}

\subsection{Tree Construction}
The construction of the hull tree is done by repeated insertions into an
initially empty tree. The overall procedure is a plane-sweep algorithm, since the
inserted points will have increasing $y$-coordinates.
We shall come back to the $\BuildTree$ routine after first looking into the $\iinsert$ algorithm.

\subsubsection{$\iinsert$.}
\label{subsection:insert}
Algorithm $\iinsert$ is a recursive algorithm. 
Rather than insert a single point at a time, we feel it is more natural to batch several such 
insertions when we can first link them into a chain.
It takes as input a  chain $C$ of vertices and a hull tree $T$.
Because $\iinsert$ will only be used in a plane-sweep manner, we will be able to
assume as a precondition that $C$ is nonempty, no point in $C$ was previously inserted, 
and no point in $C$ is dominated by the initial hull tree of $T$.

\begin{algorithm}[!ht]
    \caption{$\iinsert(C, T)$}
    \label{alg:insert}

    \SetKwInOut{Input}{Input}
    \SetKwInOut{Output}{Output}
    \SetKw{DownTo}{downto}

    \Input{$C$, a  chain of points to be inserted into $T$, \\
            $T$, the hull tree for some point set $P$.}
    \Output{$T$, the hull tree for $P \cup C$.}
    \If{$T = \nil$}{
        Create a root node with $T.hull = C$\; 
        $T.l = \head(C); T.r = \tail(C) $ \;
    }
    \Else{
        $(q_l,q_r) = \tangents(\head(C), \tail(C), T)$ \label{alg:insert:else:start} \;
        $C^\prime = $ the portion of $T.hull$ strictly between $q_l$ and $q_r$\;
        Replace $C^\prime$ by $C$ within $T.hull$ \;
        Scan and split $C^\prime$ to create these two chains\;
        \ \ \ \ \ \ $C_L=\{p\in C^\prime\:|\: p \mbox{ belongs in } T.L\}$;
        \ \ $C_R=\{p\in C^\prime\:|\: p \mbox{ belongs in } T.R\}$\;
        $T.L = \iinsert(C_L, T.L)$ \label{alg:insert:else:rec:l} ;\ \ 
        $T.R = \iinsert(C_R, T.R)$ \label{alg:insert:else:rec:r} \;\label{alg:insert:else:end}
    }
    \Return$T$ \;
\end{algorithm}

We specify the procedure $\tangents(a_l, a_r, T)$ which assumes $a_l.x <a_r.x$
$a_l.y <a_r.y$ and
neither $a_l$ nor $a_r$ is dominated by $T.hull=(h_0, h_1, \ldots , h_k, h_{k+1})$.
It returns a pair of points $(q_l,q_r)$ each from $T.hull$.
We require the line through $q_l$ and $a_l$ be a leftward tangent. 
If $h_1.x<a_l.x$ and $h_1.y<a_l.y$, this is well-defined. Otherwise, we return a degenerate
tangent with $q_l=h_0$.
Similarly, 
if $a_r.x<h_k.x$ and $a_r.y<h_k.y$, then $q_r$ defines a rightward tangent with $a_r$. 
Otherwise, we return $q_r=h_{k+1}$.

We sketch the implementation of $\tangents$.
If the leftward tangent is well-defined, we compute $q_l$ by scanning from the current position
of $T.l$.
In constant time we can determine if we should scan to the left or to the right.
(As is standard we keep track of the changing slopes of lines through $a_l$.)
Similarly, if the tangent is well-defined, we scan for $q_r$ using $T.r$.

\begin{lemma}
\label{lemma:insert:correctness}
Algorithm $\iinsert$ correctly inserts $C$ into  $T$.
\end{lemma}

\begin{proof}
The proof is by induction on the number of points.
The base case, where $T$ is empty, is clear.
In general, we only need to establish that the new $T.hull$ is correct;
this follows since all the points removed, in $C^\prime$, are dominated by the new hull.
By the recursive definition of hull trees the points of $C^\prime$ now belong in either
$T.L$ or $T.R$ and are recursively inserted into those trees. \qed
\end{proof}

\subsubsection{$\BuildTree$.}
\label{section:buildtrees}
Given a point set $P$, algorithm $\BuildTree$ starts by sorting these points by their $x$-coordinates. The 0-based index of a point $p$ in such a sorted order is called its \emph{rank}. 
As discussed above, a point's rank is used to guide its descent down the hull tree during insertion. 

\begin{algorithm}[!ht]
    \caption{$\BuildTree(P)$}
    \label{alg:BuildTree}

    \SetKwInOut{Input}{Input}
    \SetKwInOut{Output}{Output}

    \Input{$P$, a set of points, $\{p_1, p_2, \dots , p_n\}$.}
    \Output{$T$, a hull tree built from $P$.}
    Compute the rank of each point in $P$by $x$-coordinate \label{buildtree:rank}\;
    Sort the points in $P$ by increasing $y$-coordinate \;
    Create an empty hull tree $T$ \label{buildtree:newhulltree}\;
    \For{each $p$ in order}{
        $\iinsert(p, T)$ 
    } 
    \Return $T$ \;
\end{algorithm}

Recall that the $\iinsert$ procedure expects a hull chain as the first parameter, so the call to $\iinsert$ in $\BuildTree$ is a understood to be a chain of one vertex. 
Note that such singleton chains satisfy the preconditions of $\iinsert$.
Once all the points have been inserted, the hull tree is returned.

\begin{lemma}
\label{lemma:insert:tail}
Right after a point $p$ is inserted into a hull tree $T$, $\tail(T.hull) = p$.
\end{lemma}

\begin{proof}
 Since points are inserted into $T$ by increasing $y$-coordinate value, the most recently inserted point must have the largest $y$ coordinate value so far. So it must be in the root hull and cannot have any point after it. \qed
\end{proof}

\begin{lemma} 
\label{lemma:buildtree:runtime}
Algorithm $\BuildTree$ constructs a hull tree of a set of $n$ points in $\BigO{n \log n}$ time.
\end{lemma}

\begin{proof}
Clearly the initial steps are within the time bound.
It remains only to show that all the invocations of $\iinsert$ take no more than $\BigO{n \log n}$ time. 

Consider an arbitrary point $p$ inserted into $T$ by $\BuildTree$. Initially, it goes into the $T.hull$ by \cref{lemma:insert:tail}.  In subsequent iterations, the point either stays within its current hull chain or descends one level  owing to an eviction from its current hull chain. The cost of all evictions from a chain $C$ is dominated by the right-to-left tangent scan. 
We consider the number of points we scan past (i.e. not counting the points at the beginning and end of scan).
Consider the cursor $T.l$.
It scans left to right past points once; if we scan past a point a second time, going right to left,
then that point will be in $C^\prime$ and will be evicted from this level.
Symmetric observations hold for $T.r$.
And a point will be scanned a final time if it is in $C_L$ or $C_R$.
Hence we will scan past a point a constant number of times before it is evicted.

A call to $\iinsert$ takes constant time every time it is invoked (and it is
only invoked when at least one point has been evicted from its parent). In addition
$\iinsert$ takes time bounded by the number of points scanned past.
Note that any particular point $p$ starts at the root and only moves downward (when evicted) 
and there are only $O(\log n)$ levels.
Hence during the execution of $\BuildTree$ both the total number of points evicted and
the total number of points scanned past is bounded by $O(n\log n)$. \qed
\end{proof}

\begin{lemma}
\label{lemma:insert:runtime}
Each point is handled by  $\BuildTree$ in $\BigO{\log n}$ amortized time.
\end{lemma}

\begin{proof}
By \cref{lemma:buildtree:runtime}, the cost of all invocations of $\iinsert$ by 
algorithm $\BuildTree$ is $\BigO{n \log n}$, which amortizes to $\BigO{\log n}$ per point. \qed
\end{proof}

\subsection{Hull Peeling}
\label{section:hull:peeling}
We begin the discussion of hull peeling by examining algorithm $\extractHull$, which takes a valid hull tree $T$, extracts from it the root hull chain $T.hull$, and then returns it. 

\begin{algorithm}[!ht]
    \caption{$\extractHull(T)$}
    \label{alg:extractHull}

    \SetKwInOut{Input}{Input}
    \SetKwInOut{Output}{Output}

    \Input{$T$ is a hull tree for  a non-empty pointset $P$ and $H$ be the set of points in $T.hull$.}
    \Output{the hull $h$ and $T$ a hull tree for the point set $P\setminus H$}
    $h = T.hull$ \;  
    $\delete(h, T)$ \;
    \Return $h$ \;
\end{algorithm}

The correctness and cost of algorithm $\extractHull$ obviously depend on $\delete$. 
$\delete$ is called after a subchain has been cut out of the middle of the root hull chain.
This can be visualized if we imagine the root hull as a roof.
Further there is a left and a right ``overhang'' remaining after the middle of a roof has caved in.
The overhang might degenerate to just a sentinel point.
Algorithm $\delete$ itself also depends on other procedures which we discuss first.

\subsubsection{$\bbelow$.}
We could just connect the two endpoints  of the overhangs with a straight line to repair the roof.
However because of the curvature of the old roof, some of the points in $T.L$ or $T.R$
might be above this new straight line. In that case, these points need to move out of their subtrees and join in to form the 
new root hull chain.

Therefore, we will need a Boolean function $\bbelow(T,p,q)$ to that end.
It returns true if there exists a tangent of the root hull such that both $p$ and $q$ are above it.
A precondition of $\bbelow(T,p,q)$ is that the root hull can rise above the line through 
$p$ and $q$ only between $p$ and $q$.

We also specify the Boolean function $\aabove(p,q,r)$ to be true if point $r$ is above the 
line passing through $p$ and $q$. This is done with a standard constant time test.
Note that $\bbelow$ is quite different than $\aabove$.
The functions $\pred$ and $\lsucc$ operate on the corresponding hull chain in the obvious way.

\begin{algorithm}[!ht]
    \caption{$\bbelow(T, p_l, p_r)$}
    \label{alg:below}

    \SetKwInOut{Input}{Input}
    \SetKwInOut{Output}{Output}

    \Input{$T$, a hull tree, \\
           $p_l$: the rightmost end in the left overhang, \\
           $p_r$: the leftmost point in the right overhang }
    \Output{True iff every point of $T.hull$ is below the line through $p_l$ and $p_r$ }

    \lIf{$p_l$ or $p_r$ is a sentinel }{\Return false}
    \If{$\aabove(\pred(T.r), T.r, P_l)$ }{
         \While{$\lnot\aabove(T.l,\lsucc(T.l),p_r)\LAnd \lnot\aabove(p_l,p_r,T.l)$}{
        $T.l = \lsucc(T.l)$
        }
    \Return $\lnot\aabove(p_l,p_r,T.l)$ \; 
    }
    \Else{
         \While{$\lnot\aabove(\pred(T.r),T.r,p_l)\LAnd \lnot\aabove(p_l,p_r,T.r)$}{
            $T.r = \pred(T.r)$
    }
    \Return $\lnot\aabove(p_l,p_r,T.r)$ \; 
    }
\end{algorithm}

\begin{lemma}
\label{lemma:below}
Algorithm $\bbelow$ runs in linear time.
Further, if it returns false either $T.l$ or $T.r$ is above the line through $p_l$ and $p_r$.
\end{lemma}

\begin{proof}
Recall that one cursor may push the other cursor as it moves.
Only one cursor moves.
If $T.r$ is too far left to help decide, then $T.l$ moves until it is above the line or we find a tangent. \qed
\end{proof}

\subsubsection{$\getBridge$.}
Given two hull trees, where one precedes the other, algorithm $\getBridge$ scans the hull chains of the hull trees to find the bridge that connects them.

\begin{algorithm}[!ht]
    \caption{$\getBridge(T.L, T.R)$}
    \label{alg:getBridge}

    \SetKwInOut{Input}{Input}
    \SetKwInOut{Output}{Output}

    \Input{$T.L$: a left hull tree of some hull tree $T$, \\
           $T.R$: a right hull tree of $T$ }
    \Output{$p_l, p_r$: the left and right bridge points for $T.L$ and $T.R$ }
    \lIf{$ T.L = \nil$} {$ p_l = (-\infty, -\infty)$ }
    \lElse{$p_l = T.L.l$}
    
    \If{$T.R = \nil \Or \tail(T.L.hull).y > \tail(T.R.hull).y $}{$p_r = (+\infty, -\infty)$}
    \lElse{$p_r = T.R.r$ }
    
    \If{$ \aabove(T.L.l, T.R.r, \lsucc(T.L.l))$ }{
        $T.L.l = \lsucc(T.L.l)$ \;
        $(p_l,p_r) = \getBridge(T.L, T.R)$ \;
    }

    \If{$ \aabove(T.L.l, T.R.r, \pred(T.L.l))$ }{
        $T.L.l = \pred(T.L.l)$ \;
        $(p_l,p_r) = \getBridge(T.L, T.R)$ \;
    }
    
    \If{$ \aabove(T.L.l, T.R.r, \lsucc(T.R.r))$}{
        $T.R.r = \lsucc(T.R.r)$ \;
        $(p_l,p_r) = \getBridge(T.L, T.R)$ \;
    }
    
    \If{$\aabove(T.L.l, T.R.r, \pred(T.R.r))$}{
        $T.R.r = \pred(T.R.r)$ \;
        $(p_l,p_r) = \getBridge(T.L, T.R)$ \;
    }
   
    \Return $(p_l,p_r)$ \;
\end{algorithm}
                                                 
\begin{lemma}
\label{lemma:getbridge}
Given two valid hull trees $T.L$ and $T.R$, $\getBridge$ correctly computes the bridge connecting $T.L.hull$ and $T.R.hull$ in time linear in the lengths of those hulls.
\end{lemma}

\begin{proof}
The scan for the left bridge point in $T.L$ is done using its left-to-right cursor $T.L.l$,
while the scan for the right bridge point in $T.R$ is done using $T.R$'s right-to-left cursor $T.R.r$. On completion, the two cursors will be pointing to the bridge points. 
Since each vertex is scanned past at most once, the runtime is $\BigO{|T.L.hull| + |T.R.hull|)}$. This completes the proof. \qed
\end{proof}

\subsubsection{$\delete$.}
\label{subsection:delete}
The general idea of $\delete(C,T)$ is that if $C$ is a subchain of $T.hull$ then the procedure
will return with the tree $T$ being a valid hull tree for the point set $P\setminus H$, where
$H$ is the set of points in $C$.
The procedure will be employed during the peeling process, and so $C$ will initially be the 
entire root hull. The root hull will have to be replaced by moving points up from the subtrees.
In fact the points moved up will be subchains of the root hulls of $T.L$ and $T.R$. 
Recursively, these subtrees will in turn need to repair their root hulls.
We do a case analysis below and find that the only procedure we will need is $\delete$.

Again we shall employ the analogy of a roof caving in, in the middle.
The rebuilding of the ``roof'' starts with identifying the endpoints  of the remaining left and right overhangs. These points will be the sentinels if the overhangs are empty.
The endpoints $a_l$ and $a_r$ define a line segment, $(a_l, a_r)$, which we shall call the \emph{roof segment}.

Before continuing, let us examine the dynamics of the points in the root hull of a (sub)tree $T$ during successive invocations of $\delete$. Successive calls might cause the roof segment $(a_l,a_r)$ to get bigger. For successive roof segments, the root hull of $T$ is queried. There are two phases involved: before the root hull intersects the roof segment, and thereafter. During the first phase, each new roof segment is below the previous one (cf. \cref{fig:phase:1} and \cref{fig:phase:2}). During the first phase, the root hull is not changed but is queried by a series of roof segments $(a_l,a_r)$. In the second phase, it gives up its subchain from $T.l$ to $T.r$ to its parent (or is extracted). Thereafter, for each new excision,  $T.l$ and $T.r$ will move further apart, until they become a sentinel. This is shown inductively on the depth of the subtree.

In the first phase the cursors ($T.l$ and $T.r$) start at the $\head$ and $\tail$ of the list and move in response to $\bbelow$ queries. Each cursor will move in one direction at first and then, only once, change direction. This is because each subsequent query has  $(a_l,a_r)$ moving apart on the parent's convex chain. See \cref{fig:phase:1} and \cref{fig:phase:2}. Now we examine the algorithm more carefully.

\begin{algorithm}[!ht]
    \caption{$\delete(C, T)$}
    \label{alg:delete}

    \SetKwInOut{Input}{Input}
    \SetKwInOut{Output}{Output}

    \Input{$C$, a chain of points to be deleted from $T.hull$, \\
            $T$: a hull tree or points et $P$}
    \Output{$T$: The updated hull tree for the point set $P\setminus H$}

    $ a_l = \pred(\head(C));  \enspace a_r = \lsucc(\tail(C))$ \;
    $\Use_L = \lnot\bbelow(T.L,a_l,a_r); \enspace \Use_R = \lnot\bbelow(T.R,a_l,a_r) $ \label{alg:delete:call:getExtremes}\;
    \lIf {$\Use_L$}{$(L_l,L_r) = \tangents(a_l,a_r,T.L) $}
    \lIf {$\Use_R$}{$(R_l,R_r) = \tangents(a_l,a_r,T.R) $}
    \Comment{Case 1: Neither subtree used to rebuild the roof}
    \lCase{$\lnot \Use_L \LAnd \lnot \Use_R$}{
        nothing
    }
    \Comment{Case 2: Only the right subtree used to rebuild the roof}    
    \Case{$\Use_R \LAnd (\lnot \Use_L \Or \aabove(a_r,L_l,R_l))$}{
         $C_R =$ chain in $T.R.hull$ from $R_l$ to $R_r$, inclusive \;
        Update $T.hull$  with $C_R$ inserted between $a_l$ and $a_r$\;
        $\delete(C_R, T.R)$ \;
    }
    \Comment{Case 3: Only the left subtree used to rebuild the roof}    
    \Case{$\Use_L\LAnd (\lnot \Use_R \Or \aabove(R_l,a_l,L_r))$}{
        $C_L =$ chain in $T.L.hull$ from $L_l$ to $L_r$, inclusive \;
        Update $T.hull$  with $C_L$ inserted between $a_l$ and $a_r$\;
        $\delete(C_L, T.L)$ \;
    }
    \Comment{Case 4: Both subtrees used to rebuild the roof}      
    \Case{$\Use_L \LAnd \Use_R \LAnd \aabove(a_l,R_l,L_l)\LAnd\aabove(L_r,a_r,R_r)$}{
        $(q_l,q_r) = \getBridge(T.L, T.R)$ \;
        $C_L =$ chain in $T.L.hull$ from $L_l$ to $q_l$, inclusive \;
        $C_R =$ chain in $T.R.hull$ from $q_r$ to $R_r$, inclusive \;
        $D =$ chain from concatenating $C_L$ to $C_R$ \;
        Update $T.hull$  with $D$ inserted between $a_l$ and $a_r$\;
        $\delete(C_L, T.L)$ \;
        $\delete(C_R, T.R)$ \;
    }
    \Return $T$ \;

\end{algorithm}

\begin{figure}
\centering
\begin{minipage}{.5\textwidth}
  \centering
   \begin{tikzpicture}[scale=5]

\tkzDefPoint (0.000, 0.050){l}
\tkzDefPoint (0.050, 0.200){m}
\tkzDefPoint (0.1500, 0.350){al}

\tkzDefPoint (0.760, 0.660){ar}
\tkzDefPoint (0.920, 0.700){s}
\tkzDefPoint (0.999, 0.710){t}

\tkzDefPoint (0.170, 0.220){A}
\tkzDefPoint (0.190, 0.200){A2}
\tkzDefPoint (0.240, 0.350){B}
\tkzDefPoint (0.220, 0.260){B2}
\tkzDefPoint (0.380, 0.410){C}
\tkzDefPoint (0.450, 0.390){C2}
\tkzDefPoint (0.550, 0.430){D}
\tkzDefPoint (0.550, 0.400){D2}

\tkzDefPoint (0.630, 0.400){a}
\tkzDefPoint (0.660, 0.580){b}
\tkzDefPoint (0.710, 0.610){c}
\tkzDefPoint (0.900, 0.650){d}

\tkzLabelPoint[above, color=blue](al){$a_l$}
\tkzLabelPoint[above, color=blue](ar){$a_r$}

\tkzLabelPoint[left, color=red](m){$a_l'$}
\tkzLabelPoint[above, color=red](s){$a_r'$}

\tkzLabelPoint[below, color=red](A){$T.l$}
\tkzLabelPoint[below=4pt, color=red](D){$T.r$}

\tkzDrawPolySeg[style=dashed, color=blue](al, ar)
\tkzDrawPolySeg[style=dashed, color=red](m, s)
\tkzDrawPolySeg[style=solid, color=blue](A, B, C, D)
\tkzDrawPolySeg[style=dotted, color=blue](a, b, c, d)
\tkzDrawPolySeg[style=solid, color=blue](l, m, al)
\tkzDrawPolySeg[style=solid, color=blue](ar, s, t)
\tkzDrawPoints(A, B, C, D, a, b, c, d, l, m, al, s, t, ar)
\path[->] (A2) edge [bend left] (B2);
\path[->] (D2) edge [out=160, in=30] (C2);
\end{tikzpicture}
  \caption{Before $a_l$ and $a_r$ move up.}
  \label{fig:phase:1}
\end{minipage}%
\begin{minipage}{.5\textwidth}
  \centering
  \begin{tikzpicture}[scale=5]

\tkzDefPoint (0.000, 0.050){l}
\tkzDefPoint (0.050, 0.200){m}

\tkzDefPoint (0.920, 0.700){s}
\tkzDefPoint (0.999, 0.710){t}

\tkzDefPoint (0.170, 0.220){A}
\tkzDefPoint (0.200, 0.200){A2}
\tkzDefPoint (0.240, 0.350){B}
\tkzDefPoint (0.240, 0.310){B2}
\tkzDefPoint (0.380, 0.410){C}
\tkzDefPoint (0.450, 0.390){C2}
\tkzDefPoint (0.550, 0.430){D}
\tkzDefPoint (0.550, 0.400){D2}

\tkzDefPoint (0.630, 0.400){a}
\tkzDefPoint (0.660, 0.580){b}
\tkzDefPoint (0.710, 0.610){c}
\tkzDefPoint (0.900, 0.650){d}

\tkzLabelPoint[left, color=red](m){$a_l'$}
\tkzLabelPoint[above, color=red](s){$a_r'$}

\tkzLabelPoint[above=-.1pt, color=red](B){$T.l$}
\tkzLabelPoint[above, color=red](C){$T.r$}

\tkzDrawPolySeg[style=dashed, color=red](m, s)
\tkzDrawPolySeg[style=solid, color=blue](A, B, C, D)
\tkzDrawPolySeg[style=dotted, color=blue](a, b, c, d)
\tkzDrawPolySeg[style=solid, color=blue](l, m)
\tkzDrawPolySeg[style=solid, color=blue](s, t)
\tkzDrawPoints(A, B, C, D, a, b, c, d, l, m, s, t)
\path[->] (B2) edge [bend right] (A2);
\path[->] (C2) edge [out=30, in=160] (D2);

\end{tikzpicture}
  \caption{After $a_l$ and $a_r$ have moved up. }
  \label{fig:phase:2}
\end{minipage}
\end{figure}

The rebuilding process breaks into four cases depending on whether any points from
$T_L$ and $T_R$ are above the roof segment and hence will be involved in the rebuilding.

\setcounter {case} {0}
\begin{case} Neither subtree is needed to rebuild the roof.\end{case}

This case, depicted in \cref{fig:layers:delete:case1}, is when the deletion of subchain $C$ from $T.hull$ leaves a hull that already dominates all other points.

\begin{case} Only the right subtree  is needed to rebuild the roof.\end{case}
This case, depicted in \cref{fig:layers:delete:case2}, is when $T.hull$ no longer dominates the hull chain in the right subtree. 
A second subcase is when the left root hull does extend above the $(a_l,a_r)$ segment but
is still below the left tangent from the right root hull.
To maintain the hull tree invariants, a subchain of $T.R.hull$ will have to be extracted and moved up to become part of $T.hull$.
In Case 2, only the vertices of $T.R.hull$ that will be moved up are scanned past twice,
since points scanned past in phase two are removed from the current hull.

\begin{figure}[!ht]
    \begin{minipage}{.5\textwidth}
      \centering
        \begin{tikzpicture}[scale=5]

\tkzDefPoint (0.000, 0.300){l}
\tkzDefPoint (0.050, 0.460){m}
\tkzDefPoint (0.150, 0.550){n}
\tkzDefPoint (0.250, 0.600){al}

\tkzDefPoint (0.650, 0.710){ar}
\tkzDefPoint (0.800, 0.725){r}
\tkzDefPoint (0.900, 0.730){s}
\tkzDefPoint (0.990, 0.733){t}

\tkzDefPoint (0.050, 0.200){A}
\tkzDefPoint (0.100, 0.300){B}
\tkzDefPoint (0.230, 0.400){C}
\tkzDefPoint (0.400, 0.430){D}

\tkzDefPoint (0.550, 0.400){a}
\tkzDefPoint (0.600, 0.500){b}
\tkzDefPoint (0.710, 0.600){c}
\tkzDefPoint (0.810, 0.630){d}

\tkzLabelPoint[above, color=blue](al){$a_l$}
\tkzLabelPoint[above, color=blue](ar){$a_r$}
\tkzDrawPolySeg[style=dashed, color=blue](al, ar)
\tkzDrawPolySeg[style=solid, color=blue](A, B, C, D)
\tkzDrawPolySeg[style=solid, color=blue](a, b, c, d)
\tkzDrawPolySeg[style=solid, color=blue](l, m, n, al)
\tkzDrawPolySeg[style=solid, color=blue](ar, r, s, t)
\tkzDrawPoints(A, B, C, D, a, b, c, d, l, m, n, al, r, s, t, ar)

\end{tikzpicture}
      \caption{Case 1: $a_l$ can connect to $a_r$.}
      \label{fig:layers:delete:case1}
    \end{minipage}%
    \begin{minipage}{.5\textwidth}
      \centering
        \begin{tikzpicture}[scale=5]


\tkzDefPoint (0.000, 0.200){l}
\tkzDefPoint (0.050, 0.360){m}
\tkzDefPoint (0.150, 0.450){n}
\tkzDefPoint (0.250, 0.500){al}

\tkzDefPoint (0.750, 0.760){ar}
\tkzDefPoint (0.900, 0.790){s}
\tkzDefPoint (0.990, 0.793){t}

\tkzDefPoint (0.050, 0.200){A}
\tkzDefPoint (0.100, 0.300){B}
\tkzDefPoint (0.230, 0.400){C}
\tkzDefPoint (0.400, 0.430){D}

\tkzDefPoint (0.500, 0.500){a}
\tkzDefPoint (0.550, 0.680){b}
\tkzDefPoint (0.650, 0.720){c}
\tkzDefPoint (0.700, 0.730){d}
\tkzDefPoint (0.900, 0.740){e}

\tkzLabelPoint[above, color=blue](al){$a_l$}
\tkzLabelPoint[above, color=blue](ar){$a_r$}
\tkzDrawPolySeg[style=dashed, color=blue](al, ar)
\tkzDrawPolySeg[style=solid, color=blue](A, B, C, D)
\tkzDrawPolySeg[style=solid, color=blue](a, b, c, d, e)
\tkzDrawPolySeg[style=solid, color=blue](l, m, n, al)
\tkzDrawPolySeg[style=solid, color=blue](ar, s, t)
\tkzDrawPoints(A, B, C, D, a, b, c, d, e, l, m, n, al, s, t, ar)

\end{tikzpicture}
      \caption{Case 2: Right subtree involved in rebuilding.}
      \label{fig:layers:delete:case2}
  \end{minipage}
\end{figure}

\begin{case} Only the left subtree  is needed to rebuild the roof.\end{case}
This case, depicted in \cref{fig:layers:delete:case3}, is the converse of case 2,
Again, a second subcase is when the right root hull does extend above the $(a_l,a_r)$ 
segment but is still below the right tangent from the left root hull.
As in the previous case, only the vertices of $T.L.hull$ that will be moved up  are scanned twice.

\begin{case} Both subtrees are needed to rebuild the roof.\end{case}
In this case, we need to compute two subchains, one from $T.l.hull$ and the other from $T.r.hull$, which are then connected by a bridge to fix the roof. 

\begin{figure}[!ht]
    \begin{minipage}{.5\textwidth}
      \centering
        \begin{tikzpicture}[scale=5]

\tkzDefPoint (0.000, 0.050){l}
\tkzDefPoint (0.050, 0.200){m}
\tkzDefPoint (0.100, 0.300){n}
\tkzDefPoint (0.1500, 0.350){al}

\tkzDefPoint (0.750, 0.760){ar}
\tkzDefPoint (0.800, 0.785){r}
\tkzDefPoint (0.900, 0.790){s}
\tkzDefPoint (0.990, 0.793){t}

\tkzDefPoint (0.170, 0.300){A}
\tkzDefPoint (0.240, 0.430){B}
\tkzDefPoint (0.330, 0.500){C}
\tkzDefPoint (0.500, 0.530){D}

\tkzDefPoint (0.650, 0.500){a}
\tkzDefPoint (0.700, 0.680){b}
\tkzDefPoint (0.750, 0.720){c}
\tkzDefPoint (0.800, 0.730){d}
\tkzDefPoint (0.900, 0.740){e}

\tkzLabelPoint[above, color=blue](al){$a_l$}
\tkzLabelPoint[above, color=blue](ar){$a_r$}
\tkzDrawPolySeg[style=dashed, color=blue](al, ar)
\tkzDrawPolySeg[style=solid, color=blue](A, B, C, D)
\tkzDrawPolySeg[style=solid, color=blue](a, b, c, d, e)
\tkzDrawPolySeg[style=solid, color=blue](l, m, n, al)
\tkzDrawPolySeg[style=solid, color=blue](ar, r, s, t)
\tkzDrawPoints(A, B, C, D, a, b, c, d, e, l, m, n, al, r, s, t, ar)

\end{tikzpicture}
      \caption{Case 3: Only left subtree involved. }
      \label{fig:layers:delete:case3}
    \end{minipage}
    \begin{minipage}{.5\textwidth}
      \centering
        \begin{tikzpicture}[scale=5]

\tkzDefPoint (0.000, 0.050){l}
\tkzDefPoint (0.050, 0.200){m}
\tkzDefPoint (0.100, 0.300){n}
\tkzDefPoint (0.1500, 0.350){al}

\tkzDefPoint (0.750, 0.760){ar}
\tkzDefPoint (0.800, 0.785){r}
\tkzDefPoint (0.900, 0.790){s}
\tkzDefPoint (0.990, 0.793){t}

\tkzDefPoint (0.170, 0.300){A}
\tkzDefPoint (0.240, 0.430){B}
\tkzDefPoint (0.330, 0.500){C}
\tkzDefPoint (0.500, 0.530){D}

\tkzDefPoint (0.550, 0.500){a}
\tkzDefPoint (0.600, 0.680){b}
\tkzDefPoint (0.650, 0.720){c}
\tkzDefPoint (0.700, 0.730){d}
\tkzDefPoint (0.800, 0.740){e}

\tkzLabelPoint[above, color=blue](al){$a_l$}
\tkzLabelPoint[above, color=blue](ar){$a_r$}
\tkzDrawPolySeg[style=dashed, color=blue](al, ar)
\tkzDrawPolySeg[style=solid, color=blue](A, B, C, D)
\tkzDrawPolySeg[style=solid, color=blue](a, b, c, d, e)
\tkzDrawPolySeg[style=solid, color=blue](l, m, n, al)
\tkzDrawPolySeg[style=solid, color=blue](ar, r, s, t)
\tkzDrawPoints(A, B, C, D, a, b, c, d, e, l, m, n, al, r, s, t, ar)

\end{tikzpicture}
      \caption{Case 4: Both subtrees involved.}
      \label{fig:layers:delete:case4}
    \end{minipage}
\end{figure}

\begin{lemma}
\label{lemma:case4}
In Case 4, only the vertices of $T.L.hull$ and $T.R.hull$ that will be moved up to join the roof are scanned twice.
\end{lemma}
\begin{proof}
Recall that after the call to $\getBridge$, the two cursors $T.L.l$ and $T.R.r$ are already pointing to the left and right bridge points.  

The scan for the left tangent point visible to $a_l$ and above the segment $a_la_r$ is done by walking $T.L.l$ forward or backward. The decision of which direction to walk can be done in constant time. If the walk toward $T.L.r$ is chosen, then all the points encountered will be encountered for the first time. However, if the scan is backward toward the head of $T.L.hull$, then any point encountered is one that will be moved up.
A symmetrical argument applies on the right. \qed
\end{proof}

\begin{theorem}
\label{theorem:correctness:delete:cont}
Consider a sequence of calls to $\extractHull$, starting with $n$ points, until all points have been 
extracted.
The total time amortized over all calls is $\BigO{n\log n}$
\end{theorem}

\begin{proof}
We assume that we start with a hull tree (built with $\BuildTree$).
The run time is dominated by cursor movement.
Each procedure takes constant time (and the number of calls is proportional  to the
number of chain movements) plus the number of points a cursor has moved past.
The above discussion shows that each point is passed over a constant number of times before moving out a chain.
As with $\BuildTree$, this leads to our result. \qed
\end{proof}

\subsection{Merge}
Recall our discussion so far has only been for ``northwest'' hull trees, 
which we now call $T_{NW}$.
We rotate the point set 90 degrees and recompute three times, resulting in the
four hull trees $T_{NW}, T_{NE}, T_{SE}$, and $T_{SW}$.
We will use these to construct the successive convex hulls by peeling.
When the points are in general position, the extreme points 
(topmost, bottommost, rightmost, and leftmost) are on the convex hull.
Note that some of these may coincide. Further, it is clear that the chain that connects the leftmost point with the topmost is just the northwest hull chain found at the root of $T_{NW}$. The rest of the convex hull is the root hull chains of the other trees.

Initially all the points are ``unmarked''. 
When marked, a point is marked in all four trees.
We iteratively perform the following actions to construct each layer
\begin{enumerate*}[label=\emph{\alph*}), before=\unskip{: }, itemjoin={{; }}, itemjoin*={{, and }}]
\item Retrieve and delete the root hull chain from each of the hull trees
\item Remove the marked points from each chain
\item Mark the points remaining in each chain
\item Concatenate the four chains to form the convex hull for this layer.
\end{enumerate*}

This process stops when all vertices have been marked, which is when all the points
have been deleted from all the trees. This correctness follows from  above.

\begin{lemma} Given a set $S$ of $n$ points and the four hull trees of $P$ with the four orientations of $NW, NE, SE$, and $SE$, the $\merge$ procedure executes in $\BigO{n\log n}$ time.
\end{lemma}

\begin{proof}
Note that the sum of the lengths of all the chains is $\BigO{n}$. So marking points and removing them later  all in all  takes linear time. Recall that  all the calls to $\delete$ altogether take time \BigO{n\log n} time. \qed
\end{proof}

\section{Conclusion}
\label{section:layering:openProblems}
We have provided a new simple optimal algorithm for the convex layers problem. Detailed pseudocode, space and time complexity results are also given. The pseudocode might appear detailed, but that is only because the approach is simple enough that we can deal with all cases explicitly. However, by using four sets of hulls, we only need to work with monotone chains which simplifies our case analyses and makes the correctness argument straightforward.
The extension to dynamic point sets remains an open problem.

\bibliographystyle{splncs03}

\bibliography{chBiblio}

\end{document}